\documentclass[11pt]{article}
\usepackage{microtype}
\usepackage{tikz}
\usepackage[margin=1in]{geometry}
\usepackage{amssymb, amsmath, mathtools}
\usepackage{amsthm}
\usepackage[utf8]{inputenc}
\usepackage{csquotes}
\usepackage[labelfont=bf, skip=6pt]{caption}
\usepackage[short]{optidef}
\usepackage{xcolor}
\usepackage{authblk}

\newtheorem{theorem}{Theorem}
\newtheorem{lemma}[theorem]{Lemma}

\newtheorem{cor}[theorem]{Corollary}

\newtheorem{fact}[theorem]{Fact}

\theoremstyle{definition}
\newtheorem{definition}[theorem]{Definition}
\newtheorem{remark}[theorem]{Remark}

\newtheorem{notation}[theorem]{Notation}

\newcommand{\bbQ}{\mathbb{Q}}
\newcommand{\bbR}{\mathbb{R}}
\DeclareRobustCommand{\rchi}{{\mathpalette\irchi\relax}}
\newcommand{\irchi}[2]{\raisebox{\depth}{$#1\chi$}} % inner command, used by \rchi

\newcommand{\cC}{\mathcal{C}}

\newcommand{\cL}{\mathcal{L}}
\newcommand{\cP}{\mathcal{P}}
\newcommand{\cPnn}{\mathcal{P}_{n,n}}

\newcommand{\cF}{\mathcal{F}}

\title{A Real Polynomial for Bipartite Graph\\
Minimum Weight Perfect Matchings\footnote{Supported in part by NSF grant CCF-1815901.}}

\author[]{Thorben Tröbst}
\author[]{Vijay V. Vazirani}

\affil[]{\texttt{ttrbst@uci.edu, vazirani@ics.uci.edu}}
\affil[]{Department of Computer Science, University of California, Irvine}

\date{\today}

\begin{document}
    \maketitle
    
    	\begin{abstract}
In a recent paper, Beniamini and Nisan \cite{Beniamini2020bipartite} gave a closed-form formula for the unique multilinear polynomial for the Boolean function determining whether a given bipartite graph $G \subseteq K_{n,n}$ has a perfect matching, together with an efficient algorithm for computing the coefficients of the monomials of this polynomial. We give the following generalization: Given an arbitrary non-negative weight function $w$ on the edges of $K_{n,n}$, consider its set of minimum weight perfect matchings. We give the real multilinear polynomial for the Boolean function which determines if a graph $G \subseteq K_{n,n}$ contains one of these minimum weight perfect matchings.

	\end{abstract}

    \section{Introduction}
    \label{sec.intro}

    Every Boolean function $f\colon \{0, 1\}^n \rightarrow \{0, 1\}$ can be represented in a unique way as a real multilinear polynomial, and this and related representations have found numerous applications, e.g. see \cite{ODonnell}.
    In what appears to be a ground-breaking paper, Beniamini and Nisan \cite{Beniamini2020bipartite} gave this polynomial for the Boolean function determining whether a given bipartite graph $G \subseteq K_{n,n}$ has a perfect matching; importantly, they show how to efficiently compute the coefficient of any specified monomial of this polynomial. They also gave a number of applications of this fact, in particular, to communication complexity.

    Given an arbitrary non-negative weight function $w$ on the edges of $K_{n,n}$, consider its set of minimum weight perfect matchings.
    Building on the work of \cite{Beniamini2020bipartite}, we give the real multilinear  polynomial for the Boolean function which determines if a graph $G \subseteq K_{n,n}$ contains one of these minimum weight perfect matchings. As above, we can also efficiently compute the coefficients of the monomials of this polynomial.
    As mentioned in Remark \ref{rem.gen}, this is a generalization of the main theorem of \cite{Beniamini2020bipartite}. 
    
%    Our interest lies not in applications of these Real polynomials but in finding closed-form multilinear polynomials for various graph functions. Towards this end, we believe that the work of Beniamini and Nisan appears to be opening the door to interesting results of this nature. We present a number of open problems and observations to support this belief.

    \section{Key Definitions and Facts}
\label{sec.ingredients}
    
    In this section, we will give some key definitions and facts, culminating in the main theorem of \cite{Beniamini2020bipartite}. These notions pertain to bipartite graphs only.
    
    \begin{notation}
    	Given a graph $G$, $V(G)$ and $E(G)$ will denote its set of vertices and edges, respectively. If $G, H$ are two graphs, then $H$ is a {\em subgraph} of $G$, denoted $H \subseteq G$, if $V(H) = V(G)$ and $E(H) \subseteq E(G)$. $K_{n,n}$ denotes the complete balanced bipartite graph on $2n$ vertices. If $S$ is a set of subgraphs of $K_{n,n}$, then by $E(S)$ we mean the union of the sets of edges of the graphs in $S$, i.e., 
    	\[ E(S) = \bigcup_{G \in S} {E(G)} .\]
    \end{notation}
    
    \begin{definition}
    Let $\cF$ be a family of subgraphs of $K_{n,n}$ and $G$ be a subgraph of $K_{n,n}$. We will say that $G$ is {\em $\cF$-covered} if $\exists S \subseteq \cF \ s.t. \ E(G) = E(S)$. The {\em membership function for family $\cF$} is a 0/1 function $f_{\cF}$ on subgraphs of $K_{n,n}$ satisfying $f_{\cF}(G) = 1$ iff $\exists H \in \cF \ s.t. \ H \subseteq G$.
    \end{definition}
    
    Graph $G \subseteq K_{n,n}$ will be represented by the following setting of the $n^2$ variables $(x_{11}, \ldots , x_{n,n})$: $x_{i,j} = 1$ if $(i, j) \in E(G)$ and $x_{i,j} = 0$ if $(i, j) \notin E(G)$. Hence we will view the membership function for family $\cF$ as
    \[ f_{\cF}: \{0, 1 \}^{n^2} \rightarrow \{0, 1\} . \]
    
    \begin{definition}
    	Let graph $G \subseteq K_{n,n}$. By the {\em multilinear monomial corresponding to $G$}, denoted $m_G(x_{1,1}, \ldots , x_{n,n})$, we mean
    	\[ \prod_{(i, j) \in E(G)} {x_{i,j}} .\] 
    \end{definition}
    
    A following important fact was proven in \cite{Beniamini2020bipartite}:
    
    \begin{fact} (\cite{Beniamini2020bipartite})
    	\label{fact.membership}
    	  Let $\cF$ be a family of subgraphs of $K_{n,n}$.
    	The only monomials appearing in the multilinear real polynomial corresponding to the membership function for family $\cF$ are those corresponding to $\cF$-covered subgraphs of $K_{n,n}$.
    \end{fact}

    	Let $\cF$ be a family of subgraphs of $K_{n,n}$. Let $\cC(\cF)$ denote the set of all $\cF$-covered subgraphs of $K_{n,n}$ and $\hat{0}$ denote the empty subgraph of $K_{n,n}$.  The set $\cC(\cF) \cup \{ \hat{0} \}$ under the relation $\subseteq$ defines a poset. 
    	
    	\begin{fact} \cite{Beniamini2020bipartite}
    	\label{fact.poset-lattice}
    		The poset $((\cC(\cF) \cup \{ \hat{0} \}), \subseteq)$ is a lattice.
    	\end{fact}
    	
     We will denote the lattice corresponding to family $\cF$ by $\cL(\cF)$. 
    	Next we define a very special lattice. Let $\cPnn$ denote the family of perfect matchings of $K_{n,n}$. The $\cPnn$-covered graphs are called {\em matching covered graphs}, see \cite{LP.book}. The lattice $\cL(\cPnn)$ has special properties, as shown by Billera and Sarangarajan \cite{Billera1994combinatorics}:
    	
    	\begin{fact} \cite{Billera1994combinatorics}
    		\label{fact.Billera}
    		The lattice $\cL(\cPnn)$ is isomorphic to the face lattice of the Birkhoff polytope for $K_{n,n}$.
    	\end{fact} 
    	
    	It is well known that the face lattice of the Birkhoff polytope is not only graded but also Eulerian, see \cite{Stanley}. A lattice is {\em graded} if it admits a rank function $\rho$ satisfying:
    	\begin{enumerate}
    		\item If $x < y$ then $\rho(x) < \rho(y)$.
    		\item If $y$ covers $x$ (i.e., $x < y$ and there is no element $z$ in the poset such that $x < z< y$), then $\rho(y) = \rho(x) + 1$.
    	\end{enumerate}
    	Given any two elements $x < y$ in the lattice, by the {\em interval of $x$ and $y$} we mean the set $\{ z \ |\ x \leq z \leq y \}$. 
    	A graded lattice is {\em Eulerian} if in every interval, the number of even-ranked and odd-ranked elements must be equal. Since $\cL(\cPnn)$ is isomorphic to the face lattice of the Birkhoff polytope, it is also Eulerian.
    	
        \begin{definition}[see \cite{Stanley}]
    	\label{def.mobius}
    	Let $\cP = (P, <)$ be a finite partial order.
    	Then the Möbius function on $\cP$, $\mu: \cP \times \cP \rightarrow \bbR$, is defined as follows: \\
    	\[ \forall x \in P: \ \ \mu(x, x) = 1 \]
    	\[ \forall x, y \in P, \ \mbox{with} \ y < x, \ \ \mu(y, x) = - \sum_{y \leq z < x} {\mu(y, z)} \]
    	For $x \in P, \ \mu(\hat{0}, x)$ is called the {\em Möbius number} of $x$.
    \end{definition}
    
    Since $\cL(\cPnn)$ is an Eulerian lattice, one can show that the Möbius function on this lattice can be expressed in terms of the rank function as follows:
    
    \begin{fact} (see \cite{Stanley})
    	\label{fact.easy}
    	 For any graph $G \in \cL(\cPnn)$, $\mu(\hat{0}, G) = (-1)^{\rho(G)}$.
    \end{fact}
    
     Furthermore, \cite{Beniamini2020bipartite} show that $\rho(G)$ is easy to compute using the notion of cyclomatic number of a graph.

    \begin{definition}
    	\label{def.cyclomatic}
    	For any subgraph $G$ of $K_{n,n}$, define its {\em cyclomatic number} to be 
    	\[ \chi(G) = E(G) - V(G) + C(G), \]
    	where $C(G)$ is the number of connected components of $G$.
    \end{definition}
    
    The proof of Fact \ref{fact.rank} uses the notion of ear-decomposition of a connected, matching-covered bipartite graph; see \cite{LP.book} for a detailed discussion of this notion.  
    
    \begin{fact} \cite{Beniamini2020bipartite}
    	\label{fact.rank}
    	For any matching covered subgraph $G$ of $K_{n,n}$, $\rho(G) = \chi(G) + 1$.
    \end{fact}
    
    The last ingredient needed is:
    
    \begin{fact} \cite{Beniamini2020bipartite}
    	\label{fact.last}
    	    Let $\cF$ be a family of subgraphs of $K_{n,n}$. Then the membership function for this family is given by:
    	    \[  f_{\cF} (x_{1,1}, \ldots , x_{n,n}) = \sum_{G \in \cC(\cF)} { - \mu(\hat{0}, G) \cdot m_G(x_{1,1}, \ldots , x_{n,n}) } \]
    	        \end{fact}

    By Fact \ref{fact.easy}, we get the Möbius function in terms of the rank function and by Fact \ref{fact.rank} we get the rank function in terms of $\chi(G)$. Substituting this in Fact \ref{fact.last}, we get:
    
    \begin{theorem} \cite{Beniamini2020bipartite}
    	\label{thm.main-BN}
    	\[ f_{\cPnn} (x_{1,1}, \ldots , x_{n,n}) = \sum_{G \in \cC(\cPnn)} {  (-1)^{\chi(G)} \cdot m_G (x_{1,1}, \ldots , x_{n,n}) } \]
    \end{theorem}
    
        \begin{cor} \cite{Beniamini2020bipartite}
        \label{cor.odd}
        The number of matching-covered subgraphs of $K_{n,n}$ is odd.
    \end{cor}

    \begin{proof}
        Since $\cL(\cP_{n, n})$ is Eulerian, the cardinality of every closed interval is even. Furthermore, since the closed interval $[\hat{0}, K_{n,n}]$ contains all matching-covered subgraphs of $K_{n,n}$ and the empty graph, the lemma follows.
    \end{proof}
    
   \begin{remark}
   Fact \ref{fact.last} gives the multilinear real polynomial in closed form for an arbitrary  membership function; however, it is in terms of the Möbius function. It turns out that the latter is not easy to compute in general, as we will see in Section \ref{sec.other}. The key step taken by \cite{Beniamini2020bipartite} was to show how to efficiently compute the Möbius number corresponding to any specified monomial of formula for the membership function of perfect matchings. This critically depends on the fact that the lattice $\cPnn$ is Eulerian, which follows from Fact \ref{fact.Billera}, and Fact \ref{fact.easy}, which relates the rank of a matching covered graph to its cyclomatic number. Of course the entire formula will in general have exponentially many monomials.
      \end{remark}

    \section{Extension to Minimum Weight Perfect Matchings}

   We next turn to minimum weight perfect matchings.
Let $w : E(K_{n, n}) \rightarrow \bbQ_+$ be a fixed but arbitrary weight function.

    \begin{definition}
        Let $\cP^w_{n, n}$ denote the family of minimum weight perfect matchings of $K_{n,n}$ with respect to weight function $w$, and $f_{n,w}: \{0, 1\}^{n^2} \rightarrow \{0, 1\}$ denote the membership function for this family. Let $\cC(\cPnn^w)$ denote the set of all $\cPnn^w$-covered subgraphs of $K_{n,n}$.
           \end{definition}

By Fact \ref{fact.poset-lattice}, the poset $((\cC(\cPnn^w) \cup \{ \hat{0} \}), \subseteq)$ is a lattice, which we will denote by $\cL(\cPnn^w)$.

    \begin{lemma}\label{lem:downward_hull}
        Let $G_w$ be the graph whose edge set is the union of all matchings in $\cP^w_{n, n}$.
        Then
        \[
            \cL(\cP^w_{n, n}) \subseteq \{G \in \cL(\cP_{n, n}) \mid G \leq G_w\},
        \]
        i.e.\ $\cL(\cP^w_{n, n})$ is the closed interval $[0, G_w]$ in $\cL(\cP_{n, n})$.
    \end{lemma}

    \begin{proof}
        This follows from the fact that if $M \subseteq G_w$ is any perfect matching, then $M \in \cP^w_{n, n}$.
        To see this note that the face of the optimum solutions of the LP
        \begin{mini*}
            {}
            {w(x)}
            {}
            {}
            \addConstraint{x(\delta(a))}{= 1}{\quad \forall a \in A}
            \addConstraint{x(\delta(b))}{= 1}{\quad \forall b \in B}
            \addConstraint{x_e}{\geq 0}{\quad \forall e \in E(K_{n, n})}
        \end{mini*}
        is given by the polytope
        \[
            Q \coloneqq \left\{x \in \bbR^{E(K_{n, n})} \;\middle|\; \begin{array}{lll}
                    x(\delta(a)) &= 1 & \forall a \in A, \\
                    x(\delta(b)) &= 1 & \forall b \in B, \\
                    x_e &= 0 & \forall e \notin E(G_w), \\
                    x_e &\geq 0 & \forall e \in E.
            \end{array}\right\}
        \]
        But any perfect matching $M \subseteq G_w$ defines a point in $Q$ which means that $w(M)$ is minimum.
        Thus $M \in \cP^w_{n, n}$.
    \end{proof}

    \begin{theorem}
    \label{thm:min_weight_bpm}
    
            	\[ f_{n,w} (x_{1,1}, \ldots , x_{n,n}) = \sum_{G \in \cC(\cPnn^w)} {  (-1)^{\chi(G)} \cdot m_G (x_{1,1}, \ldots , x_{n,n}) } \]

    \end{theorem}

    \begin{proof}
        By Fact~\ref{fact.last} we have
        \[
            f_{n, w}(x) = \sum_{G \in \cC(\cP^w_{n, n})}{- \mu(\hat{0}, G) \cdot m_G(x)}.
        \]
        But since $\cL(\cP^w_{n, n})$ is just the downward hull of $G_w$ in $\cL(\cP_{n, n})$, we know that $\mu_{\cL(\cP^w_{n, n})} = \mu_{\cL(\cP_{n, n})}$.
        Thus
        \begin{align*}
            f_{n, w}(x) &= \sum_{G \in \cC(\cP^w_{n, n})}{- \mu_{\cL(\cP_{n, n})}(\hat{0}, G) \cdot m_G(x)} \\
                        &= \sum_{G \in \cC(\cP^w_{n, n})}{{(-1)}^{\rchi(G)} \cdot m_G(x)}. \qedhere
        \end{align*}
    \end{proof}

    \begin{remark}
    	\label{rem.gen}
    	As in \cite{Beniamini2020bipartite}, the coefficients of the monomials of this polynomial can also be efficiently computed. Analogous to Corollary \ref{cor.odd}, the number of $\cP^w_{n, n}$-covered subgraphs of $K_{n,n}$ is odd; this follows for the fact that $\cL(\cP_{n, n})$ is Eulerian. Finally, observe that Theorem \ref{thm:min_weight_bpm} is a generalization of Theorem \ref{thm.main-BN} since for the special case of unit weights on all edges of $K_{n,n}$, all perfect matchings have the same weight.
    \end{remark}

    \section{Discussion}
    \label{sec.other}

It is easy to check that most of the machinery of \cite{Beniamini2020bipartite} extends readily to non-bipartite graphs. However, their main theorem does not extend in any straightforward manner. Let $\cP_n$ denote the set of perfect matchings of $K_n$ and $\cL(\cP_n)$ denote the lattice of matching covered subgraphs of $K_n$. This lattice is not Eulerian, and not even graded. Therefore new structural properties are needed for obtaining an efficient algorithm for computing the associated Möbius numbers. The possibility that this function is NP-hard to compute is also not ruled out; however, it is important to point out that over the years, for numerous algorithmic results, the non-bipartite case has followed the bipartite case, with the infusion of appropriate structural facts, which are typically quite non-trivial; see Section 1.3 in \cite{Anari2020matching} for an extensive discussion of this phenomenon. 

More generally, it will be interesting to study the complexity of computing individual Möbius numbers of other lattices which have polynomial representations. An important candidate is the lattice of stable matchings for which a succinct representation follows from Birkhoff's representation theorem \cite{Birkhoff} and the notion of rotations, see \cite{GusfieldI}. 

 We note that the smallest value of $n$ for which $\cL(\cP_n)$ is not Eulerian is $n = 6$.
    In particular, for the graph $G$ (see Figure~\ref{fig:p6_not_eulerian_1}), the interval $[0, G]$ (see Figure~\ref{fig:p6_not_eulerian_2}) is not Eulerian.
    Note also that $\mu(0, G) = 0$ which means that the monomial $m_G(x)$ does not appear in the associated graph covering polynomial. Indeed, one can show that $\cL(\cP_6)$ is not even graded. We exhibit an explicit pentagon sublattice in Figure~\ref{fig:p6_not_graded}.

    \begin{figure}[htb]
        \begin{center}
            \begin{tikzpicture}[scale=2]
                \node[circle, draw, inner sep=2pt, thick] (v0) at (0, 0) {};
                \node[circle, draw, inner sep=2pt, thick] (v1) at (3, 0) {};
                \node[circle, draw, inner sep=2pt, thick] (v2) at (3, 2) {};
                \node[circle, draw, inner sep=2pt, thick] (v3) at (0, 2) {};
                \node[circle, draw, inner sep=2pt, thick] (v4) at (1, 1) {};
                \node[circle, draw, inner sep=2pt, thick] (v5) at (2, 1) {};

                \draw[-] (v0) -- (v1) -- (v2) -- (v3) -- (v0);
                \draw[-] (v0) -- (v4) -- (v3);
                \draw[-] (v1) -- (v5) -- (v2);
                \draw[-] (v4) -- (v5);
            \end{tikzpicture}
        \end{center}
        \caption{The graph $G \subseteq K_6$ referred to in Figure \ref{fig:p6_not_eulerian_2}.
        \label{fig:p6_not_eulerian_1}}
    \end{figure}
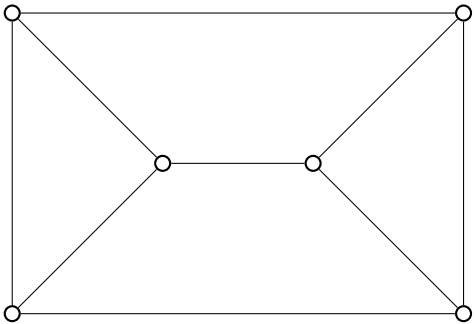

    \begin{figure}[htb]
        \begin{center}
            \begin{tikzpicture}[scale=1.5]
                \node[circle, draw, inner sep=2pt, thick] (b) at (0, 0) {};

                \node[circle, draw, inner sep=2pt, thick] (l11) at (-1.5, 1) {};
                \node[circle, draw, inner sep=2pt, thick] (l12) at (-0.5, 1) {};
                \node[circle, draw, inner sep=2pt, thick] (l13) at (0.5, 1) {};
                \node[circle, draw, inner sep=2pt, thick] (l14) at (1.5, 1) {};

                \node[circle, draw, inner sep=2pt, thick] (l21) at (-2.5, 2) {};
                \node[circle, draw, inner sep=2pt, thick] (l22) at (-1.5, 2) {};
                \node[circle, draw, inner sep=2pt, thick] (l23) at (-0.5, 2) {};
                \node[circle, draw, inner sep=2pt, thick] (l24) at (0.5, 2) {};
                \node[circle, draw, inner sep=2pt, thick] (l25) at (1.5, 2) {};
                \node[circle, draw, inner sep=2pt, thick] (l26) at (2.5, 2) {};

                \node[circle, draw, inner sep=2pt, thick] (l31) at (-1, 3) {};
                \node[circle, draw, inner sep=2pt, thick] (l32) at (0, 3) {};
                \node[circle, draw, inner sep=2pt, thick] (l33) at (1, 3) {};

                \node[circle, draw, inner sep=2pt, thick] (t) at (0, 4) {};

                \draw[thick, ->] (b) -- (l11);
                \draw[thick, ->] (b) -- (l12);
                \draw[thick, ->] (b) -- (l13);
                \draw[thick, ->] (b) -- (l14);

                \draw[thick, ->] (l11) -- (l22);
                \draw[thick, ->] (l11) -- (l24);
                \draw[thick, ->] (l11) -- (l25);

                \draw[thick, ->] (l12) -- (l23);
                \draw[thick, ->] (l12) -- (l24);
                \draw[thick, ->] (l12) -- (l26);

                \draw[thick, ->] (l13) -- (l21);
                \draw[thick, ->] (l13) -- (l25);
                \draw[thick, ->] (l13) -- (l26);

                \draw[thick, ->] (l14) -- (l21);
                \draw[thick, ->] (l14) -- (l22);
                \draw[thick, ->] (l14) -- (l23);

                \draw[thick, ->] (l21) -- (l31);

                \draw[thick, ->] (l22) -- (l31);
                \draw[thick, ->] (l22) -- (l32);

                \draw[thick, ->] (l23) -- (l32);

                \draw[thick, ->] (l24) -- (l32);
                \draw[thick, ->] (l24) -- (l33);

                \draw[thick, ->] (l25) -- (l31);
                \draw[thick, ->] (l25) -- (l33);

                \draw[thick, ->] (l26) -- (l33);

                \draw[thick, ->] (l31) -- (t);
                \draw[thick, ->] (l32) -- (t);
                \draw[thick, ->] (l33) -- (t);
            \end{tikzpicture}
        \end{center}
        \caption{Interval $[0, G] \subseteq \cL(\cP_6)$ for the graph of Figure~\ref{fig:p6_not_eulerian_1}, showing that $\cL(\cP_6)$ is not Eulerian. 
        \label{fig:p6_not_eulerian_2}}
    \end{figure}
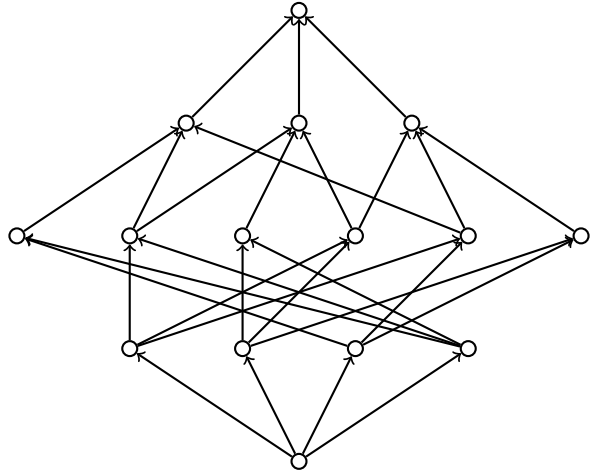

%    This lack of structure makes it difficult to compute the Möbius function and indeed it may be computationally hard to do so.
 %   One can define the sequence $A(n) = \mu(0, K_{2n})$ for $n \geq 0$.
%    It begins $0, 1, 1, -9$ (\TODO: I would really like to compute another term) which could potentially match sequence $A296171(n)$ from the On-Line Encyclopedia of Integer Sequences (\TODO: reference).

    \begin{figure}[htb]
        \begin{center}
            \begin{tikzpicture}[scale=1.5]
                \node[circle, draw, dotted, inner sep=1.3cm] (c1) at (0, 0) {};
                \node[circle, draw, dotted, inner sep=1.3cm] (c2) at (3, -3) {};
                \node[circle, draw, dotted, inner sep=1.3cm] (c3) at (3, -6) {};
                \node[circle, draw, dotted, inner sep=1.3cm] (c4) at (-3, -4.5) {};
                \node[circle, draw, dotted, inner sep=1.3cm] (c5) at (0, -9) {};

                \draw[thick, ->] (c5) -- (c3);
                \draw[thick, ->] (c5) -- (c4);
                \draw[thick, ->] (c3) -- (c2);
                \draw[thick, ->] (c4) -- (c1);
                \draw[thick, ->] (c2) -- (c1);

                \begin{scope}
                    \node[circle, draw, inner sep=2pt, thick] (v0) at (0:1) {};
                    \node[circle, draw, inner sep=2pt, thick] (v1) at (60:1) {};
                    \node[circle, draw, inner sep=2pt, thick] (v2) at (120:1) {};
                    \node[circle, draw, inner sep=2pt, thick] (v3) at (180:1) {};
                    \node[circle, draw, inner sep=2pt, thick] (v4) at (240:1) {};
                    \node[circle, draw, inner sep=2pt, thick] (v5) at (300:1) {};

                    \draw[-] (v0) -- (v1);
                    \draw[-] (v0) -- (v4);
                    \draw[-] (v0) -- (v5);
                    \draw[-] (v1) -- (v2);
                    \draw[-] (v1) -- (v3);
                    \draw[-] (v1) -- (v5);
                    \draw[-] (v2) -- (v3);
                    \draw[-] (v2) -- (v4);
                    \draw[-] (v2) -- (v5);
                    \draw[-] (v3) -- (v4);
                    \draw[-] (v4) -- (v5);
                \end{scope}

                \begin{scope}[shift={(3, -3)}]
                    \node[circle, draw, inner sep=2pt, thick] (v0) at (0:1) {};
                    \node[circle, draw, inner sep=2pt, thick] (v1) at (60:1) {};
                    \node[circle, draw, inner sep=2pt, thick] (v2) at (120:1) {};
                    \node[circle, draw, inner sep=2pt, thick] (v3) at (180:1) {};
                    \node[circle, draw, inner sep=2pt, thick] (v4) at (240:1) {};
                    \node[circle, draw, inner sep=2pt, thick] (v5) at (300:1) {};

                    \draw[-] (v0) -- (v1);
                    \draw[-] (v0) -- (v4);
                    \draw[-] (v0) -- (v5);
                    \draw[-] (v1) -- (v2);
                    \draw[-] (v1) -- (v3);
                    \draw[-] (v1) -- (v5);
                    \draw[-] (v2) -- (v3);
                    \draw[-] (v2) -- (v5);
                    \draw[-] (v3) -- (v4);
                    \draw[-] (v4) -- (v5);
                \end{scope}

                \begin{scope}[shift={(3, -6)}]
                    \node[circle, draw, inner sep=2pt, thick] (v0) at (0:1) {};
                    \node[circle, draw, inner sep=2pt, thick] (v1) at (60:1) {};
                    \node[circle, draw, inner sep=2pt, thick] (v2) at (120:1) {};
                    \node[circle, draw, inner sep=2pt, thick] (v3) at (180:1) {};
                    \node[circle, draw, inner sep=2pt, thick] (v4) at (240:1) {};
                    \node[circle, draw, inner sep=2pt, thick] (v5) at (300:1) {};

                    \draw[-] (v0) -- (v1);
                    \draw[-] (v0) -- (v4);
                    \draw[-] (v0) -- (v5);
                    \draw[-] (v1) -- (v2);
                    \draw[-] (v1) -- (v5);
                    \draw[-] (v2) -- (v3);
                    \draw[-] (v2) -- (v5);
                    \draw[-] (v3) -- (v4);
                    \draw[-] (v4) -- (v5);
                \end{scope}

                \begin{scope}[shift={(-3, -4.5)}]
                    \node[circle, draw, inner sep=2pt, thick] (v0) at (0:1) {};
                    \node[circle, draw, inner sep=2pt, thick] (v1) at (60:1) {};
                    \node[circle, draw, inner sep=2pt, thick] (v2) at (120:1) {};
                    \node[circle, draw, inner sep=2pt, thick] (v3) at (180:1) {};
                    \node[circle, draw, inner sep=2pt, thick] (v4) at (240:1) {};
                    \node[circle, draw, inner sep=2pt, thick] (v5) at (300:1) {};

                    \draw[-] (v0) -- (v1);
                    \draw[-] (v0) -- (v4);
                    \draw[-] (v0) -- (v5);
                    \draw[-] (v1) -- (v2);
                    \draw[-] (v1) -- (v3);
                    \draw[-] (v1) -- (v5);
                    \draw[-] (v2) -- (v3);
                    \draw[-] (v2) -- (v4);
                    \draw[-] (v3) -- (v4);
                    \draw[-] (v4) -- (v5);
                \end{scope}

                \begin{scope}[shift={(0, -9)}]
                    \node[circle, draw, inner sep=2pt, thick] (v0) at (0:1) {};
                    \node[circle, draw, inner sep=2pt, thick] (v1) at (60:1) {};
                    \node[circle, draw, inner sep=2pt, thick] (v2) at (120:1) {};
                    \node[circle, draw, inner sep=2pt, thick] (v3) at (180:1) {};
                    \node[circle, draw, inner sep=2pt, thick] (v4) at (240:1) {};
                    \node[circle, draw, inner sep=2pt, thick] (v5) at (300:1) {};

                    \draw[-] (v0) -- (v1);
                    \draw[-] (v0) -- (v4);
                    \draw[-] (v0) -- (v5);
                    \draw[-] (v1) -- (v2);
                    \draw[-] (v1) -- (v5);
                    \draw[-] (v2) -- (v3);
                    \draw[-] (v3) -- (v4);
                    \draw[-] (v4) -- (v5);
                \end{scope}
            \end{tikzpicture}
        \end{center}
        \caption{A pentagon sublattice in the lattice $\cL(\cP_6)$, showing that the latter is not graded.\label{fig:p6_not_graded}}
    \end{figure}
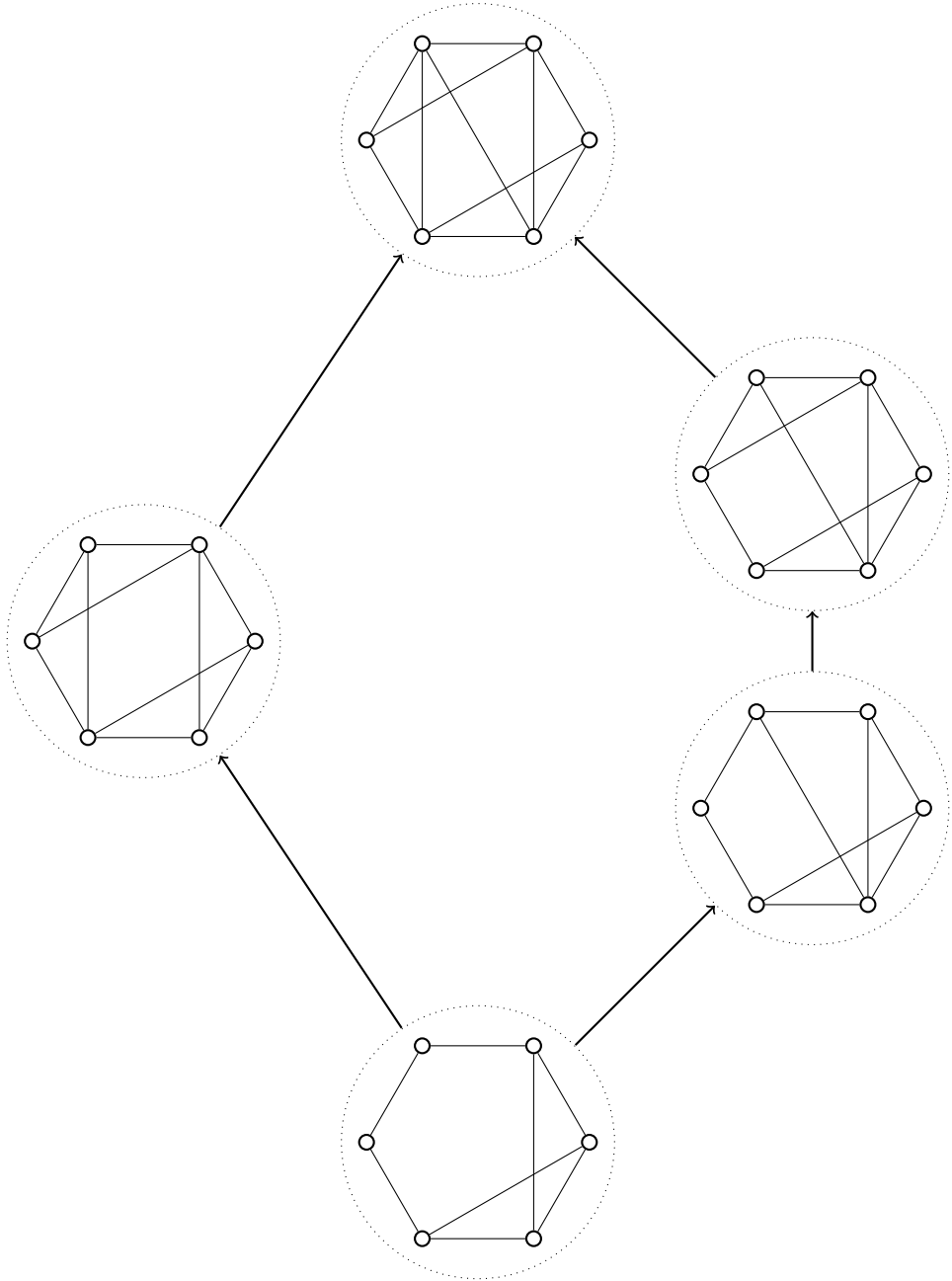
    
    Besides non-bipartite perfect matching, one can seek real polynomials for several other graph properties, including connected subgraphs; out-branchings from a given root vertex in a directed graph; and subgraphs containing an $s$-$t$ path, for given vertices $s$ and $t$.

    \bibliographystyle{alpha}
	\bibliography{refs}
\end{document}